\documentclass[11pt]{article}
\usepackage{amsfonts,epsfig,amssymb,amsmath,a4wide}

\newcommand{\MSOA}{\mbox{MSO}_1}

\newtheorem{theorem}{Theorem}[section]

\newtheorem{corollary}[theorem]{Corollary}
\newtheorem{lemma}[theorem]{Lemma}

\newtheorem{alg}[theorem]{Algorithm}
\newenvironment{proof}{\noindent{\bf Proof~}}{\null\hfill $\Box$\par\medskip}

\begin{document}

\title{A note on module-composed graphs}

\author{Frank Gurski\thanks{Heinrich-Heine Universit\"at D\"usseldorf,
Department of Computer Science, D-40225 D\"usseldorf, Germany,
E-Mail: gurski-corr@acs.uni-duesseldorf.de, }}

\bigskip

\maketitle

\begin{abstract}
In this paper we consider module-composed graphs, i.e. graphs which can be defined by a sequence of one-vertex insertions $v_1,\ldots,v_n$, such
that the neighbourhood of vertex $v_i$, $2\leq i\leq n$, forms a module (a homogeneous set) of the graph defined by vertices $v_1,\ldots,v_{i-1}$.

We show that
module-composed graphs are HHDS-free and thus homogeneously orderable, weakly chordal, and perfect. Every bipartite distance hereditary graph, every (co-$2C_4,P_4$)-free graph and thus every trivially perfect graph is module-composed. We give an $O(|V_G|\cdot (|V_G|+|E_G|))$ time algorithm to decide whether a given graph $G$ is module-composed and construct a corresponding 
module-sequence.

For the case of bipartite graphs, module-composed graphs
are exactly distance hereditary graphs, which implies
simple linear time algorithms for their recognition and construction of a corresponding module-sequence.

\end{abstract}

\section{Preliminaries}

Let $G=(V_G,E_G)$ be a graph.
For some vertex $v\in V_G$ we denote the {\em neighbourhood} of $v$ by
$N(v)=\{w\in V_G~|~\{v,w\}\in E_G\}$. $M\subseteq V_G$ is called a {\em module (homogeneous set)} of $G$, if
and only if for all $(v_1,v_2)\in M^2$: $N(v_1)-M=N(v_2)-M$, i.e. $v_1$ and $v_2$ have
identical neighbourhoods outside $M$. $M\subseteq V_G$ is called a {\em trivial module},
if $|M|=0$,  $|M|=1$, or $M=V_G$, see \cite{CH94}. A graph $G$ is called {\em prime}
if every module of $G$ is trivial.
A module $M$ is {\em maximal} if there is no non-trivial module $N$ such that $M\subseteq N$.
A module is called {\em strong} if it does not overlap
with any other module. 

While the set of modules of a graph $G$ can be exponentially large, the set of strong modules is linear
in the number of vertices. The inclusion order of the set of all strong modules
defines a tree-structure which  is denoted as {\em modular decomposition} $T_G$, see \cite{MR84}.
The root of $T_G$ represents the graph $G$ and the leaves of $T_G$ correspond to
the vertices of $G$. Every inner node, i.e. non-leaf node, $w$ of $T_G$ corresponds to an induced subgraph
of $G$ consisting of the leaves of $T_G$ in subtree with root $w$, which is called the {\em representative graph} of $w$ 
and is denoted by $G(w)$. Vertex set
$V_{G(w)}$ is a strong module of $G$. For some inner node $v$ of $T_G$, the {\em quotient graph}  $G[v]$ is
obtained by substituting in $G(v)$ every strong module, represented by some child of $v$ in $T_G$, by a
single vertex. For some inner node $v$ of  $T_G$, quotient graph $G[v]$ is either an independent set
($v$ is denoted as {\em co-join node}), a clique ($v$ is denoted as {\em join node}), or a prime graph ($v$ is denoted as
{\em prime node}).

For $U\subseteq V_G$, we define by
$G[U]$ the subgraph of $G$ induced by the vertices of $U$.
For some graph $G$, we denote its edge complement by co-$G$. For a set of graphs ${\mathcal F}$, we denote by {\em ${\mathcal F}$-free graphs} the set of all graphs that do not contain a graph of ${\mathcal F}$
as an induced subgraph.

In Table \ref{gr} we show some special graphs to which we refer during the paper.
A {\em hole} is a chordless cycle with at least five vertices. 
A {\em $k$-sun} is a chordal graph $G$ on $2k$ vertices for some $k\ge 3$ whose vertex
set can be partitioned into  $V_G=U\cup W$ such that $U=\{u_0,\ldots,u_{k-1}\}$ and $W=\{w_0,\ldots,w_{k-1}\}$ 
is an independent set. Additionally vertex $u_i$
is adjacent to vertex $w_j$ if and only if $i=j$ or $i=j+1 \mod k$.
$G$ is called a $sun$ if it is a $k$-sun for some $k\ge 3$.
If graph  $G[U]$ is a clique, then $G$ is called a {\em complete $k$-sun}.

\begin{table}[ht]
\begin{center}
\begin{tabular}{ccccccc}  
\\
\epsfig{figure=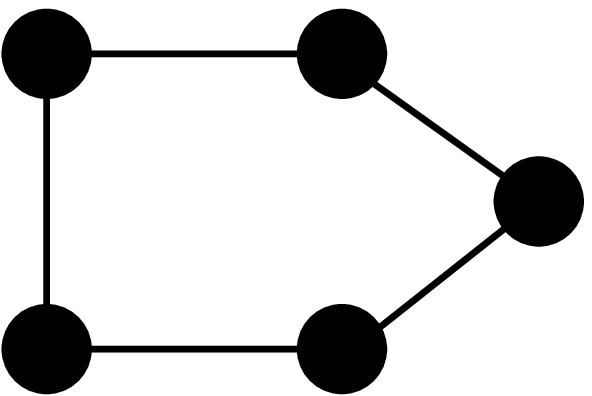,width=1.5cm}  && \epsfig{figure=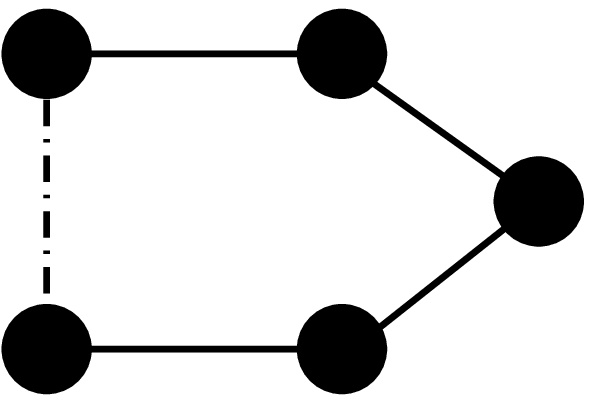,width=1.5cm} && \epsfig{figure=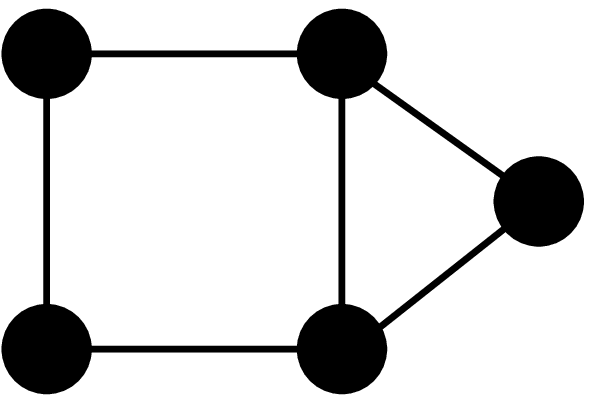,width=1.5cm} &&   \epsfig{figure=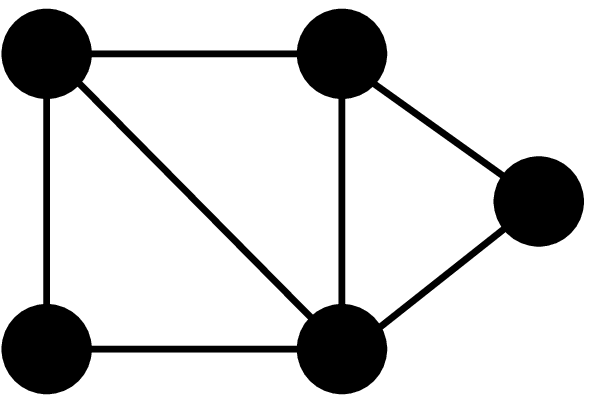,width=1.5cm}        \\
$C_5$ &&  hole && house && gem  \\ \\
\epsfig{figure=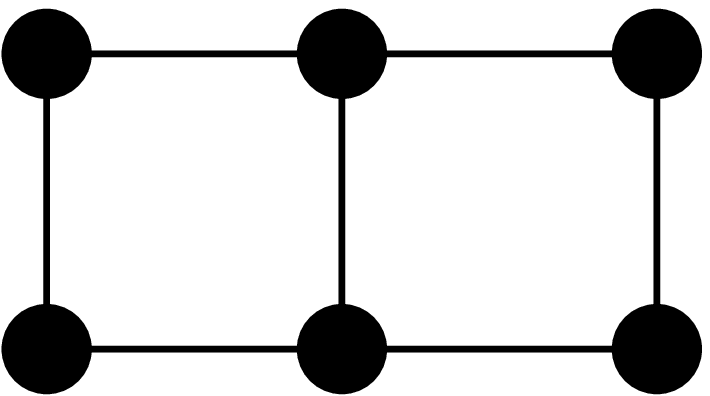,width=1.8cm} && \epsfig{figure=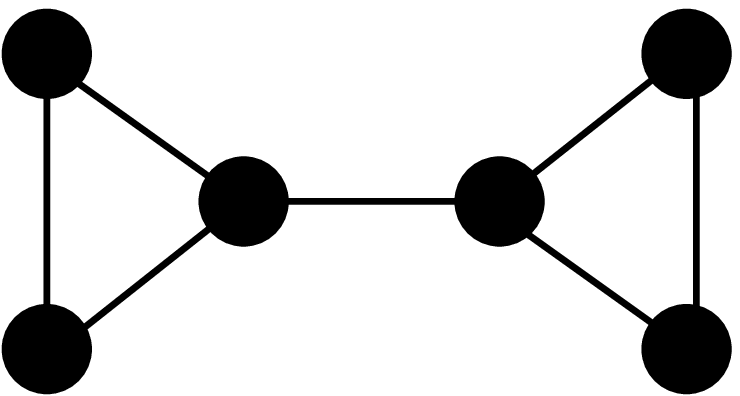,width=1.8cm}  && \epsfig{figure=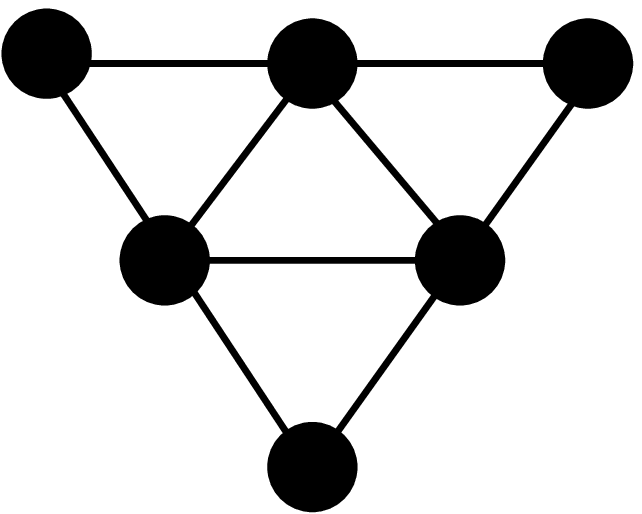,width=1.6cm} && \epsfig{figure=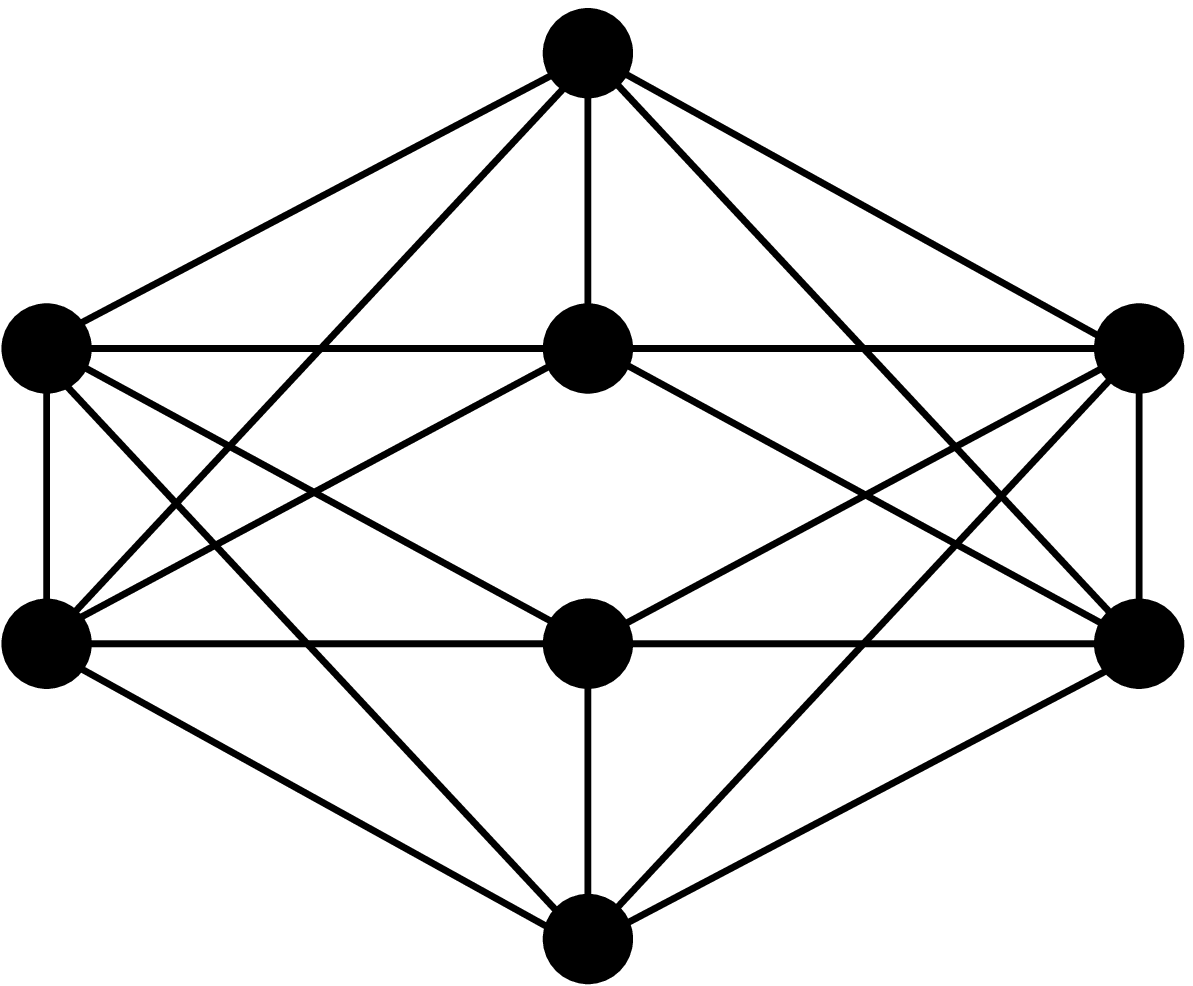,width=2.8cm}\\
domino &&  co-($K_{3,3}-e$)       &&  3-sun  && co-$2C_4$ 
\end{tabular}
\end{center}
\caption
{Special graphs}
\label{gr}
\end{table}

\section{Module-composed graphs}

There are several graph classes which are defined by a sequence of one-vertex
extentions of restricted form. Some well known examples are trees, co-graphs,
and distance hereditary graphs, see \cite{Rao07} for a survey. We next analyze
a closely related but new concept.

Graph $G$ is {\em module-composed}, if and only if there exists a linear 
ordering $\varphi:V_G \to [|V_G|]$, such that for every $2\leq i\leq|V_G|$ the 
neighbourhood of vertex $\varphi^{-1}(i)$ in graph 
$G[\{\varphi^{-1}(1),\ldots,\varphi^{-1}(i-1)\}]$ forms a module. 
For some module-composed graph $G$, $\varphi$ is called a {\em module-sequence} for $G$.

The definition of module-composed graphs
was introduced \cite{AGKKW06} for computing
connectivity ratings for vertices in special graph classes, 
see also \cite{AKKW06}.
We first recall the following easy but important lemma from \cite{AGKKW06}.

\begin{lemma}[Induced subgraph]\label{Li}
If a graph $G$ is module-composed, then every induced subgraph of $G$ is also 
module-composed.
\end{lemma}

Given two module-sequences $\varphi_1,\varphi_2$ for two graphs $G_1$ and $G_2$, 
sequence $\varphi(v)=\varphi_1(v), v\in V_{G_1}$ and $\varphi(v)=\varphi_2(v)+|V_{G_1}|, v\in V_{G_2}$ is a possible module-sequence for the disjoint union of these two graphs.

\begin{lemma}[Disjoint union]\label{Ld}
For two module-composed graphs $G_1,G_2$, the disjoint union  $G_1\cup G_2$ is also 
module-composed.
\end{lemma}

The following observation follows from Lemma \ref{Li} and the definition of 
module-composed graphs.

\begin{lemma} \label{Lx}
A graph $G$ is  module-composed, if and only if there exists a
vertex $v\in V_G$ such that $N(v)$ is a module in graph $G[V_G-\{v\}]$ and graph $G[V_G-\{v\}]$ 
is module-composed.
\end{lemma}

By Lemma \ref{Lx} the following graphs (see Table \ref{gr}) are  not module-compo\-sed,
since none of them contains a vertex $v$ such that $N(v)$ is a module in graph $G[V_G-\{v\}]$:

$C_n$, $n\ge 5$ (i.e. holes), co-$C_n$, $n\ge 5$ (i.e. anti-holes), house,  domino,  
co-($K_{3,3}-e$), 3-sun, co-$2C_4$.

The example of graph co-$2C_4$ shows that not every co-graph\footnote{A {\em co-graph} is either a  single vertex $\bullet$,
the disjoint union $G_1 \cup G_2$ of two co-graphs $G_1,G_2$, or the join  $G_1 \times G_2$ of two co-graphs $G_1,G_2$, which connects every vertex of $G_1$ with every vertex of $G_2$. } is module-compo\-sed. Graph co-$2C_4$ can even be used to characterize those co-graphs which are module-composed.

\begin{lemma}\label{Lco}
Let $G$ be a co-graph. The following conditions are equivalent.
\begin{enumerate}
\item $G$ is module-composed.

\item $G$ is (co-$2C_4$)-free.
\end{enumerate}
\end{lemma}

\begin{proof}
If $G$ is module-composed, then by Lemma \ref{Li} it obviously contains no co-$2C_4$ as induced subgraph.

Let $G$ be (co-$2C_4$)-free co-graph. Then there exists a co-graph expression $X$ defined by the three co-graph operations (single vertex $\bullet$,
disjoint union $G_1 \cup G_2$ of two co-graphs $G_1,G_2$, join  $G_1 \times G_2$ of two co-graphs $G_1,G_2$) for $G$.  Any subexpression $\bullet$ and  $G_1 \cup G_2$
are also feasible for a module-sequence. 

Let $X'=X_1 \times X_2$
be a subexpression of $X$. Since the graph defined by $X'$ contains no co-$2C_4$ as an induced subgraph either graph defined by $X_1$ or that by $X_2$ defines a subgraph of $K_1\cup K_2$, i.e. the disjoint union of a clique 
on two vertices and a clique on one vertex. Let us assume that  $X_2$ does so.  This allows us to define
a module decomposition for $X$ as follows. We start with
a module-sequence for $X_1$, which exists by induction, proceed with
the vertices of $K_2$ and finish with vertex of graph $K_1$,
which leads a module-sequence for graph defined by $X$.
\end{proof}

Co-graphs are exactly $P_4$-free graphs which implies our next corollary.

\begin{corollary}
(co-$2C_4,P_4$)-free graphs are module-composed.
\end{corollary}

Further it is  known that trivially perfect\footnote{
A graph is {\em trivially perfect} if for every induced subgraph $H$ of $G$,
the size of the largest independent set in $H$ equals the number of  all maximal cliques in $H$.} graphs are exactly  $(C_4,P_4)$-free graphs
\cite{Gol78}, which obviously form a subclass of $($co-$2C_4,P_4)$-free graphs.

\begin{corollary}
Trivially perfect graphs are module-composed.
\end{corollary}

%
%

Next we conclude results on super classes
of module-composed graphs.

It is easy to see that the house, every hole and the domino are not module-composed.
By a result shown in \cite{Far83} each sun contains a complete sun as induced subgraph, which
is obviously not module-composed. By  Lemma \ref{Li} the next result follows.

\begin{lemma}
Module-composed graphs are HHDS-free\footnote{(house,hole,domino,sun)-free}.
\end{lemma}

Since HHDS-free graphs are perfect\footnote{A graph $G$ is {\em perfect} if, for every induced subgraph $H$ of $G$, the chromatic number of $H$ is equal to the size of a maximum clique of $H$.}, the same holds true for module-composed graphs.

\begin{corollary}
Module-composed graphs are perfect.
\end{corollary}

Further, HHDS-free graphs are homogeneously orderable  by the results shown in \cite{BDN97}, which implies the same for module-composed graphs.

\begin{corollary}
Module-composed graphs are homogeneously orderable.
\end{corollary}

Since the graph $C_4$ is module-composed but not chordal, we conclude that 
module-composed graphs are not chordal, but they are weakly chordal\footnote{A graph is {\em weakly chordal} if it does not 
contain any induced cycles of length greater than four or their complements.}, since they
are HHD-free\footnote{(house,hole,domino)-free} and HHD-free graphs are weakly chordal.

\begin{corollary}
Module-composed graphs are weakly chordal.
\end{corollary}

\section{Algorithms for module-composed graphs}

Next we give a polynomial time algorithm to recognize module-composed graphs.
Our algorithm is based on Lemma \ref{Lx}. In order to find some vertex $v$ that satisfies the 
conditions of Lemma \ref{Lx}, we use a modular decomposition \cite{CH94} 
in our following Algorithm \ref{a2}. A basic observation is that for every connected module-composed
graph $G$ vertex $v$ is either a child or a grandchild of the root of $T_G$.

\begin{alg} \label{a2} ~
\hrule

\medskip
\noindent
Input: Graph $G$

\noindent
Output: Module-sequence $\varphi: V_G \to [|V_G|]$ or the answer NO

\begin{tabbing}
(11) \= d \= c \=  d \=  d \=  d \=  d\kill
(1)  \>  mod-com($G$)   \\
(2)  \> \> if ($G$ disconnected)\\
(3)  \> \> \> for every connected component $H$ of $G$: mod-com($H$); \\
(4)  \> \> else \{ \\
(5)  \> \> \> construct $T_G$ with root $r$; \\
(6)  \> \> \> if ($r$ is join node) \{ \\
(7)  \> \> \> \> if ($\exists$ child $v_l$ of $r$ which is a leaf in $T_G$) \{  \\
(8)  \> \> \> \>    \>for every such child $v_l$ of $r$  \{$\varphi(v_l)=i++$; $G=G-\{v_l\}$; \}\\
(9)  \> \> \> \> \>mod-com($G$);\}\\
(10)  \> \> \> \> else if ($\exists$ child $r_1$ of $r$ labeled by co-join and a child $v_l$ of $r_1$ which \\
(11)  \> \> \> \> is a leaf in $T_G$) \{  \\
(12)  \> \> \> \>    \>for every such vertex $v_l$  \{$\varphi(v_l)=i++$; $G=G-\{v_l\}$; \}\\
(13) \> \> \> \> \>mod-com($G$); \} \\ 
(14)  \> \> \> \} \\
(15)  \> \> \> else if ($r$ is prime node) \{ \\
(16)  \> \> \> \> if ($\exists$ child $v_1$ of $r$ which is a leaf in $T_G$ and corresponds
to a vertex \\
(17)  \> \> \> \> of degree 1 in quotient graph $G[r]$) \{  \\
(18)  \> \> \> \>    \>for every such child $v_1$ of $r$  \{$\varphi(v_1)=i++$; $G=G-\{v_1\}$; \}\\
(19)  \> \> \> \> \>mod-com($G$);\}\\
(20)  \> \> \> \> else if ($\exists$ child $r_1$ of $r$ labeled by co-join and corresponds
to a vertex  \\
(21)  \> \> \> \> of degree 1 in quotient graph $G[r]$ and a child $v_1$ of $r_1$ which is a \\
(22)  \> \> \> \> leaf in $T_G$) \{  \\
(23)  \> \> \> \> \>for every such vertex $v_1$  \{$\varphi(v_1)=i++$; $G=G-\{v_1\}$; \}\\
(24)  \> \> \> \> \>mod-com($G$); \} \\ 
(25)  \> \> \> \} \\
(26)  \> \> \> else\\
(27)  \> \> \> \> return NO;\\
(28)  \> \> \} \\
\end{tabbing}

\vspace{-0.5cm}
\hrule
\end{alg}

The construction of the modular decomposition $T_G$ in Line (5) of Algorithm \ref{a2} 
can be realized in time $O(|V_G|+|E_G|)$ by \cite{CH94,MS99}.

\begin{theorem}
Given a graph $G$, one can decide in time $O(|V_G|\cdot(|V_G| + |E_G|))$ whether
$G$ is module-composed, and in the case of a positive answer, constructs  a module-sequence.
\end{theorem}

Since module-composed graphs are HHD-free, we conclude by the results shown in
\cite{JO88} the following theorem.

\begin{theorem}
For every module-composed graph which is given together with a module-sequence the size of a largest independent set, the size of a largest clique, 
the chromatic number and the minimum number of cliques covering the graph
can be computed in linear time.
\end{theorem}

\section{Independent module-composed graphs}

Next we want to characterize module-composed graphs for a restricted case.

A graph $G$ is {\em independent module-composed}, if and only if there exists a linear 
ordering $\varphi:V_G \to [|V_G|]$, such that for every $2\leq i\leq|V_G|$ 
the neighbourhood of vertex $\varphi^{-1}(i)$ in graph $G[\{\varphi^{-1}(1),\ldots,\varphi^{-1}(i-1)\}]$ 
forms a module which is an independent set.

It is easy to see that independent module-composed graphs do not contain
any of the graphs of Table \ref{gr} as induced subgraph.

\begin{lemma}
Independent module-composed graphs are HHDG-free\footnote{(house,hole,domino,gem)-free}.
\end{lemma}

HHDG-free are also known as distance hereditary graphs \cite{HM90,BM86}.
Examples for distance hereditary graphs are co-graphs and trees.
For the case of bipartite graphs\footnote{A graph is bipartite if it is $C_{2n+1}$-free, for $n\ge 1$.}, the notion module-composed 
even is equivalent to the notion of distance hereditary. 

\begin{theorem}[\cite{AGKKW06}]\label{bipm}  Let $G$ a bipartite graph. The following conditions are equivalent.
\begin{enumerate}
\item $G$ is module-composed.

\item $G$ is domino and hole free.

\item $G$ is distance hereditary.

\item $G$ is   $(6,2)$-chordal\footnote{ A graph is $(k,l)$-chordal if each cycle
of length at least $k$ has at least $l$ chords.}.
\end{enumerate}
\end{theorem}

For general graphs Theorem \ref{bipm} does not hold true, since 
there are module-composed graphs which are not distance hereditary, e.g. the gem and there
are  distance hereditary graph which are not module-composed, e.g. the
co-($K_{3,3}-e$).


\bigskip
The problem to decide whether a given graph is bipartite distance hereditary
and to construct a corresponding pruning sequence can be done in linear time
by the well known characterization for bipartite graphs as 2-colorable graphs and
the linear time recognition algorithms for distance hereditary graphs shown in \cite{HM90,BM86}.
By Theorem \ref{bipm}, this immediately implies a linear time algorithms for recognizing
independent module-composed graphs. A corresponding module-sequence can be constructed 
in linear time from a pruning sequence as shown in \cite{AGKKW06}.
Since both known linear time recognition algorithms for distance hereditary graphs shown 
in \cite{HM90,BM86} are based on the fact that the neighbourhood of every vertex in 
a  distance hereditary graph is a co-graph and additional conditions, both algorithms 
are not simple.

In \cite{JO88} it is shown that for HHD-free graphs every Lex-BFS (Lexicographic Breadth First Search) 
ordering is a semi perfect elimination ordering, i.e. every vertex $\varphi^{-1}(i)$ is no midpoint  of an
induced $P_4$ in graph 
$G[\{\varphi^{-1}(1),\ldots,\varphi^{-1}(i-1)\}]$. In the case of bipartite graphs 
this ordering obviously is even an independent module-sequence.

\begin{theorem}
Given an independent module-composed graph $G$, every  Lex-BFS ordering  
constructs in  time $O(|V_G|+|E_G|)$ an independent module-sequence for $G$.
\end{theorem}

To decide whether a given graph is bipartite distance hereditary can be
done by Corollary 5 shown in \cite{BM86} using the fundamental search strategy 
of BFS (Breadth First Search) which produces a classification of the vertices into levels, 
with respect to a start vertex $u$. Level $i$ is the set of vertices
with distance $i$ to vertex $u$ and is denoted by $N_i(u)$.

\begin{theorem}[Corollary 5 of \cite{BM86}]
Let $G$ be a connected graph and let $u$ be a vertex of $G$. Then $G$
is bipartite distance hereditary if and only if all levels $N_k(u)$
are edgeless, and for every vertices $v$,$w$ in $N_k(u)$ and neighbours
$x$ and $y$ of $v$ in $N_{k-1}(u)$, we have $N(x)\cap N_{k-2}(u)=N(y)\cap N_{k-2}(u)$,
and further $N(v)\cap N_{k-1}(u)$ and $N(w)\cap N_{k-1}(u)$ are either disjoint
or one is contained in the other.
\end{theorem}

A BFS starting at a vertex $u$ can compute the level sets $N_k(u)$ 
in time $O(|V_G|+|E_G|)$ and using these levels,
the conditions of Corollary 5 of \cite{BM86} can be verified in the same time.

A BFS numbering $\varphi$ of the vertices with respect to some
vertex $u$ can be used to obtain a
module-sequence $\varphi_1$ as follows. We start with $\varphi_1(v)=\varphi(v)$, $\forall v\in V_G$.
For the first $|N_0(u)|+|N_1(u)|$ vertices we obviously 
can choose $\varphi_1(v)=\varphi(v)$. For the vertices of $w\in N_k(u)$, $k\ge 2$,
we know that their neighbours in set $N_{k-1}(u)$ are modules which can be ordered
by a series of inclusions $N^1\subseteq N^2 \subseteq \ldots \subseteq N^j$. We rearrange the order of 
the vertices in $N_k(u)$ with respect to $\varphi_1$ such that for every such series of inclusions 
$\varphi_1(w_1)<\varphi_1(w_2)$ if and only if $N_{k-1}(u)\cap N(w_1)\supseteq N_{k-1}(u)\cap N(w_2)$.
This obviously leads a module-sequence for graph $G$ if $G$ is bipartite distance hereditary.

\begin{theorem}
Given a graph $G$, one can decide using  BFS in time $O(|V_G|+|E_G|)$ whether
$G$ is independent module-composed, and in the case of a positive answer, 
construct a module-sequence.
\end{theorem}

On bipartite distance hereditary graphs, and so on independent module-composed graphs, the 
path-partition problem \cite{YC98},
hamiltonian circuit and path problem \cite{MN93},
and the computation of shapley value ratings  \cite{AGKKW06}
can be solved in polynomial time.

It is well known that distance hereditary graphs and
thus independent module-composed graphs have clique-width at most 3 \cite{GR00}. This implies that
all graph properties which are expressible in monadic second order logic with
quantifications over vertices and vertex sets ($\MSOA$-logic) are decidable in
linear time on independent module-composed graphs \cite{CMR00}. Some of these problems are
partition into $k$ independent sets or cliques, $k$-dominating set, $k$-achromatic
number, for every fixed integer $k$.

Furthermore, there are a lot of NP-complete graph problems which are not expressible in
$\MSOA$-logic like chromatic number, partition problems,  vertex disjoint paths, and
bounded degree subgraph problems but which can also be solved in polynomial time on clique-width
bounded graphs and thus on bipartite distance hereditary graphs \cite{EGW01a,GW06}. 

Note that general module-composed graphs are of unbounded clique-width. 
For example every graph which can be constructed from a single
vertex by a sequence of one vertex extentions by a domination vertex\footnote{A vertex $v\in V_G$ is a {\em dominating vertex} of $G$, if it is adjacent to all other vertices in $G$.} or a
pendant vertex\footnote{A vertex $v\in V_G$ of degree one is called a {\em pendant vertex} of $G$.} is 
obviously module-composed. But the set of all such defined graphs
have unbounded clique-width  \cite{Rao07}.

\section{Graph class inclusions}

In Table \ref{grcl} we summarize the relation of module-composed graphs
and related graph classes. For the definition and relations of special graph classes
we refer to the survey of Brandst\"adt  et al. \cite{BLS99}.

\begin{table}
\begin{center}
\epsfig{figure=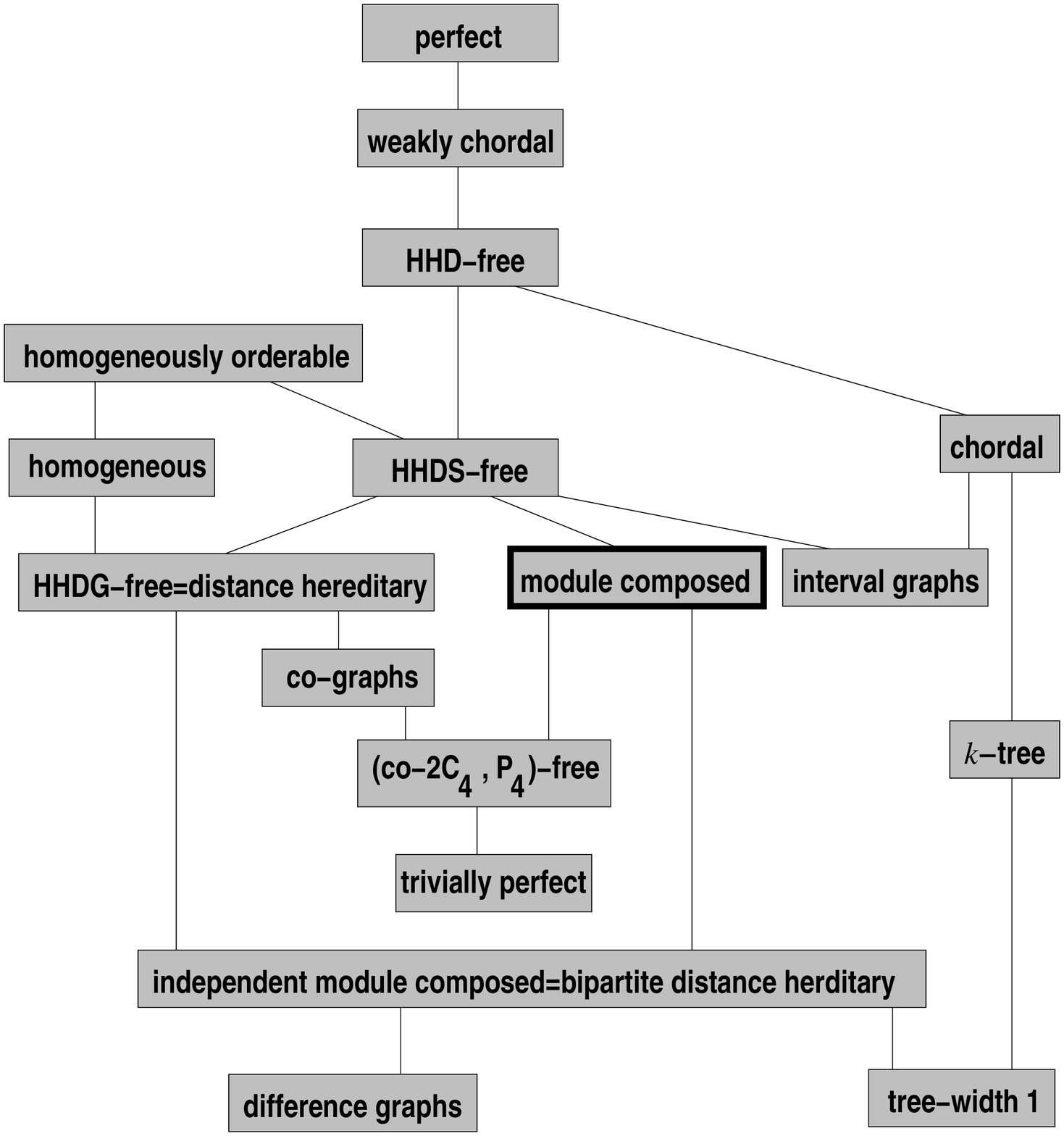,width=11.0cm}
\caption{Inclusion of special graph classes}\label{grcl}
\end{center}
\end{table}


\bibliographystyle{alpha}
\bibliography{/home/gurski/bib.bib}

\newcommand{\etalchar}[1]{$^{#1}$}
\begin{thebibliography}{AKKW06}

\bibitem[AGK{\etalchar{+}}06]{AGKKW06}
M.~Abraham, F.~Gurski, A.~Krumnack, R.~K\"otter, and E.~Wanke.
\newblock A connectivity rating for vertices in networks.
\newblock submitted, 2006.

\bibitem[AKKW06]{AKKW06}
M.~Abraham, R.~K\"otter, A.~Krumnack, and E.~Wanke.
\newblock A connectivity rating for vertices in networks.
\newblock In {\em Proceedings of the 4th IFIP International Conference on
  Theoretical Computer Science-TCS}, pages 283--298. Springer, 2006.

\bibitem[BDN97]{BDN97}
A.~Brandst\"adt, F.~F. Dragan, and F.~Nicolai.
\newblock Homogeneously orderable graphs.
\newblock {\em Theoretical Computer Science}, 172:209--232, 1997.

\bibitem[BLS99]{BLS99}
A.~Brandst\"adt, V.B. Le, and J.P. Spinrad.
\newblock {\em Graph Classes: A Survey}.
\newblock SIAM Monographs on Discrete Mathematics and Applications. SIAM,
  Philadelphia, 1999.

\bibitem[BM86]{BM86}
H.-J. Bandelt and H.M. Mulder.
\newblock Distance-hereditary graphs.
\newblock {\em Journal of Combinatorial Theory, Series B}, 41:182--208, 1986.

\bibitem[CH94]{CH94}
A.~Cournier and M.~Habib.
\newblock A new linear time algorithm for modular decomposition.
\newblock In {\em Procedings of CAAP}, volume 787 of {\em LNCS}, pages 68--84.
  Springer, 1994.

\bibitem[CMR00]{CMR00}
B.~Courcelle, J.A. Makowsky, and U.~Rotics.
\newblock Linear time solvable optimization problems on graphs of bounded
  clique-width.
\newblock {\em Theory of Computing Systems}, 33(2):125--150, 2000.

\bibitem[EGW01]{EGW01a}
W.~Espelage, F.~Gurski, and E.~Wanke.
\newblock How to solve {NP}-hard graph problems on clique-width bounded graphs
  in polynomial time.
\newblock In {\em Proceedings of Graph-Theoretical Concepts in Computer
  Science}, volume 2204 of {\em LNCS}, pages 117--128. Springer, 2001.

\bibitem[Far83]{Far83}
M.~Farber.
\newblock Characterizations of strongly chordal graphs.
\newblock {\em Discrete Mathematics}, 43:173--189, 1983.

\bibitem[Gol78]{Gol78}
M.C. Golumbic.
\newblock Trivially perfect graphs.
\newblock {\em Discrete Mathematics}, 24:105--107, 1978.

\bibitem[GR00]{GR00}
M.C. Golumbic and U.~Rotics.
\newblock On the clique-width of some perfect graph classes.
\newblock {\em International Journal of Foundations of Computer Science},
  11(3):423--443, 2000.

\bibitem[GW06]{GW06}
F.~Gurski and E.~Wanke.
\newblock Vertex disjoint paths on clique-width bounded graphs.
\newblock {\em Theoretical Computer Science}, 359(1-3):188--199, 2006.

\bibitem[HM90]{HM90}
P.L. Hammer and F.~Maffray.
\newblock Completely separable graphs.
\newblock {\em Discrete Applied Mathematics}, 27:85--99, 1990.

\bibitem[JO88]{JO88}
B.~Jamison and S.~Olariu.
\newblock On the semi-perfect elimination.
\newblock {\em Advances in applied mathematics}, 9:364--376, 1988.

\bibitem[MN93]{MN93}
H.~M\"uller and F.~Nicolai.
\newblock Polynomial time algorithms for hamiltonian problems on bipartite
  distance-hereditary graphs.
\newblock {\em Information Processing Letters}, 46(5):225--230, 1993.

\bibitem[MR84]{MR84}
R.H. M{\"o}hring and F.J. Radermacher.
\newblock Substitution decomposition for discrete structures and connections
  with combinatorial optimization.
\newblock {\em Annals of Discrete Mathematics}, 19:257--365, 1984.

\bibitem[MS99]{MS99}
R.M. McConnell and J.~Spinrad.
\newblock Modular decomposition and transitive orientation.
\newblock {\em Discrete Mathematics}, 201(1-3):189--241, 1999.

\bibitem[Rao07]{Rao07}
M.~Rao.
\newblock Clique-width of graphs defined by one-vertex extensions.
\newblock Manuscript, 2007.

\bibitem[YC98]{YC98}
H.-G. Yeh and G.J. Chang.
\newblock The path-partition problem in bipartite distance-hereditary graphs.
\newblock {\em Taiwanese Journal of Mathematics}, 2(3):353--360, 1998.

\end{thebibliography}

\end{document}